%% file: main.tex
\documentclass{llncs}
\usepackage[utf8]{inputenc}
\usepackage{amsmath}
\usepackage{amssymb}
\usepackage{stmaryrd}
\usepackage[inference]{semantic}
\usepackage{listings}
\newcommand{\eval}{\mathit{eval}}
\newcommand{\justif}{\vdash}

\newcommand{\ppo}{\leq}
\newcommand{\eppo}{\mathit{eppo}}
\newcommand{\condo}{\mathit{condo}}
\newcommand{\ctrlpo}{\mathit{ctrlpo}}
\newcommand{\ctrlocs}{\mathir{ctrlocs}}
\newcommand{\conflict}{\mathit{\#}}
\newcommand{\mrdlk}{\textsc{lk}}
\newcommand{\mrdrf}{\textsc{rf}}
\newcommand{\mrddep}{\textsc{dp}}
\newcommand{\support}{\textit{support}}
\newcommand{\domain}{\textsf{dom}}
\newcommand{\range}{\textsf{ran}}
\newcommand{\cloc}{\textit{cloc}}
\newcommand{\dep}{\mrddep}

\newcommand{\ltsarrow}[2]{\stackrel{#2}{\longrightarrow}_{#1}}

\newcommand{\futrel}[1]{\ltsarrow{{\sf F}}{#1}}
\newcommand{\fullrel}[1]{\ltsarrow{}{#1}}

\newcommand{\future}{\mathop{\triangleright}}
\newcommand{\restr}[2]{#1_{|#2}}

\usepackage{xspace}
\newcommand{\MRD}{{\sf MRD}\xspace}

\newcommand{\blue}[1]{{\color{blue}#1}}

\newcommand{\rffun}{\gamma}
\newcommand{\pset}[1]{\overline{#1}}
\newcommand{\sdef}{\mathrel{\widehat{=}}}
\newcommand{\sconfigc}{\mathbb{C}}
\newcommand{\mrd}{\preprocess}

\newcommand{\preprocess}{\mathcal{P}}

\newcommand\evW[2]{{\rm W}\,{#1}\,{#2}}
\newcommand\evR[2]{{\rm R}\,{#1}\,{#2}}
\newcommand\evWR[2]{{\rm W}^R\,{#1}\,{#2}}
\newcommand\evRA[2]{{\rm R}^A\,{#1}\,{#2}}
\newcommand\evWRO[2]{{\rm W}^{[R]}\,{#1}\,{#2}}
\newcommand\evRAO[2]{{\rm R}^{[A]}\,{#1}\,{#2}}

\newcommand\levW[3]{\text{#1:}\evW{#2}{#3}}
\newcommand\levR[3]{\text{#1:}\evR{#2}{#3}}
\newcommand\levWR[3]{\text{#1:}\evWR{#2}{#3}}
\newcommand\levRA[3]{\text{#1:}\evRA{#2}{#3}}
\newcommand\levWRO[3]{\text{#1:}\evWRO{#2}{#3}}
\newcommand\levRAO[3]{\text{#1:}\evRAO{#2}{#3}}

\usepackage{tikz}
\usepackage{forest}

\newcommand{\reffig}[1]{Fig.~\ref{#1}}

\newcommand{\reflem}[1]{Lem\-ma~\ref{#1}}

\newcommand{\refsec}[1]{Section~\ref{#1}}
\newcommand{\refex}[1]{Example~\ref{#1}}
\newcommand{\refdef}[1]{Definition~\ref{#1}}

\newcommand{\OW}{\mathit{O\!W}\!}

\newcommand{\wrval}{{\it wrval}}
\newcommand{\loc}{{\it var}}

\newcommand{\observedWrites}{\mathit{E\!W}\!}

\newcommand{\TA}{\mathit{TA}}

\newcommand{\R}{\mathsf{Rd}}

\newcommand{\W}{\mathsf{Wr}}

\newcommand{\tid}{{\it tid}}

\newcommand{\OMIT}[1]{}

\newcounter{sarrow}
\newcommand\strans[1]{%
  \mathrel{\raisebox{0.1em}{
    \stepcounter{sarrow}%
    \!\!\!\!
    \begin{tikzpicture}
      \node[inner sep=.5ex] (\thesarrow) {$\; \scriptstyle #1 \;$};
      \path[draw,<-,decorate,line width=0.25mm,
      decoration={zigzag,amplitude=0.7pt,segment length=1.2mm,pre=lineto,pre length=4pt}] 
      (\thesarrow.south east) -- (\thesarrow.south west);
    \end{tikzpicture}%
  }}}

\newcommand{\ltsArrow}[1]{\stackrel{#1}{\Longrightarrow}}

\tikzset{
    mo/.style={solid,->,>=stealth,thick,black!20!purple},
    hb/.style={solid,->,>=stealth,thick,blue},
    sw/.style={solid,->,>=stealth,thick,black!50!green},
    rf/.style={solid,->,>=stealth,thick,black!50!green},
    fr/.style={dashed,->,>=stealth,thick,red}
 }

\newcommand{\kwskip}{\textsf{\textbf{skip}}}
\newcommand{\kwdo}{\textsf{\textbf{do}}}
\newcommand{\kwwhile}{\textsf{\textbf{while}}}

\newcommand{\kwif}{\textsf{\textbf{if}}}
\newcommand{\kwthen}{\textsf{\textbf{then}}}
\newcommand{\kwelse}{\textsf{\textbf{else}}}

\newcommand{\ltsb}{{\sf sb}}
\newcommand{\lteco}{{\sf eco}}
\newcommand{\ltrf}{\mathord{\sf rf}}
\newcommand{\ltfr}{{\sf fr}}
\newcommand{\lthb}{{\sf hb}}

\newcommand{\ltmo}{{\sf mo}}

\newcommand{\True}{{\it true}}
\newcommand{\False}{{\it false}}

\newcommand{\Comm}{{\it Com}}
\newcommand{\AComm}{{\it ACom}}

\newcommand{\whilestep}[1]{\stackrel{#1}{\longrightarrow}}

\newcommand{\sem}[1]{\llbracket #1 \rrbracket}

\usetikzlibrary{decorations}
\usetikzlibrary{decorations.pathmorphing}
\usetikzlibrary{decorations.pathreplacing}
\usetikzlibrary{fit,calc,backgrounds,arrows,arrows.meta,patterns,decorations}
\usetikzlibrary{matrix}
\usetikzlibrary{positioning}
\usetikzlibrary{shapes}

\tikzstyle{tid}=[font=\bf, align=center]
\tikzstyle{rid}=[font=\sc, align=center]
\tikzstyle{tsep}=[thin,double]
\tikzstyle{code}=[font=\ttfamily\small, align=center]

\tikzstyle{basic} = [rectangle, rounded corners=3pt, thin,draw=black, fill=blue!8, sibling distance=20mm, minimum width=11mm, minimum height=5.5mm,inner sep=0pt]
\tikzstyle{bold} = [rectangle, rounded corners=3pt, thin, draw=black, fill=orange!45, sibling distance=20mm]
\tikzstyle{invisible} = [draw=none]
\tikzstyle{po}=[->,thick,color=black]
\tikzstyle{dp}=[->,thick,color=black!20!yellow]

\tikzstyle{conf} = [-, draw=black!30!red, line width=0.9, decorate, decoration={zigzag, segment length=7, amplitude=2.6, pre=lineto, pre length=0pt, post=lineto, post length=0pt}]
\tikzstyle{just} = [->, draw=green!50!black,line width=1.000, dashed]
\tikzstyle{lab} = [xshift=0.5cm,yshift=-0.1cm,color=green!20!black,font=\scriptsize]
\tikzstyle{labright} = [xshift=0.5cm,yshift=-0.1cm,color=green!20!black,font=\scriptsize]
\tikzstyle{lableft} = [xshift=-0.5cm, yshift=-0.1cm,color=green!20!black,font=\scriptsize]

\tikzstyle{poloc}=[->,thick,color=purple]
\tikzstyle{co}=[->,thick,color=brown!50!black]
\tikzstyle{rf}=[->,thick,color=red]
\tikzstyle{fr}=[->,thick,dashed,color=green!50!black]
\tikzstyle{sc}=[->,thick,cyan!50!black]
\tikzstyle{vppo}=[->,thick,purple!70!blue]
\tikzstyle{cause}=[->,thick,blue!50!black]
\tikzstyle{si}=[-,thick,double,brown!70!black]
\tikzstyle{sync}=[->,thick,color=cyan!50!black,]
\tikzstyle{bob}=[->,thick,color=cyan!50!black,]
\tikzstyle{addr}=[->,thick,color=cyan!50!black,]
\tikzstyle{strongrf}=[->,thick,color=red]

\usepackage{pifont}

\newcommand{\imp}{\Rightarrow}
\newcommand{\mrdlab}{{\it Lab}}

\usepackage{cite}

\title{Owicki-Gries Reasoning for C11 Programs with Relaxed Dependencies (Extended Version)\thanks{We thank Simon Doherty for discussions on an earlier version of this work. Wright is supported by VeTSS.      Batty is supported by EPSRC grants EP/V000470/1 and EP/R032971/1, and the Royal Academy of Engineering. Dongol is supported by EPSRC grants EP/V038915/1, EP/R032556/1, EP/R025134/2, VeTSS and ARC Discovery Grant DP190102142. }}

\author{Daniel Wright\inst{1} \and Mark Batty\inst{1} \and Brijesh Dongol\inst{2}}
\institute{University of Kent \and  University of Surrey}

\begin{document}

\maketitle

\begin{abstract}
    Deductive verification techniques for C11 programs have advanced significantly in recent years with the development of operational semantics and associated logics for increasingly large fragments of C11. However, these semantics and logics have been developed in a restricted setting to avoid the {\em thin-air-read} problem. 
    In this paper, we propose an operational semantics that leverages an intra-thread partial order (called {\em semantic dependencies}) induced by a recently developed denotational event-structure-based semantics. We prove that our operational semantics is sound and complete with respect to the denotational semantics. We present an associated logic that generalises a recent Owicki-Gries framework for RC11 (repaired C11), and demonstrate the use of this logic over several example proofs. 
\end{abstract}

\input{introduction}

\input{background}

\input{ppopp}

\input{opsem}

\input{hoare-logic}

\input{example}

\input{conclusions}

\bibliographystyle{splncs04}
\bibliography{references}

\newpage
\appendix
\input{appendix1}

\input{appendix}

\input{appendix-example}



\end{document}

%% file: introduction.tex
\section{Introduction}

Significant advances have now been made on the semantics of (weak) memory models using a variety of axiomatic (aka declarative), operational and denotational techniques. Several recent works have therefore focussed on logics and associated verification frameworks for {\em reasoning} about program executions over weak memory models. These include specialised separation logics \cite{DBLP:conf/esop/SvendsenPDLV18,DBLP:journals/pacmpl/DangJKD20,DBLP:conf/esop/DokoV17,DBLP:conf/ecoop/KaiserDDLV17} and adaptations of classical Owicki-Gries reasoning \cite{DBLP:conf/ecoop/DalvandiDDW19,LahavV15}. At the level of languages, there has been a particular focus on C11 (the 2011 C/C++ standard). 

Due to the complexity of C11 \cite{DBLP:conf/popl/BattyOSSW11}, many reasoning techniques have restricted themselves to particular fragments of the language by only allowing certain types of memory accesses and/or reordering behaviours. Several works (e.g., \cite{DBLP:conf/ecoop/KaiserDDLV17,LahavV15,DBLP:conf/ppopp/DohertyDWD19,DBLP:conf/ecoop/DalvandiDDW19}) assume a memory model that guarantees that program order (aka sequenced-before order) is maintained~\cite{DBLP:conf/pldi/LahavVKHD17}. While this restriction makes the logics and proofs of program correctness more manageable, it precludes reasoning about observable executions displaying certain real-world phenomena, e.g., those exhibited by the load-buffering litmus test under Power or ARMv7. 

\smallskip
\noindent {\it Load buffering and the thin air problem.}
C11 suffers from the \emph{out of thin air} problem~\cite{DBLP:conf/esop/BattyMNPS15}, where the language does not impose any ordering between a read and the writes that depend on its value. The intention is to universally permit aggressive compiler optimisation, but in doing so, this accidentally allows writes to take illogical values. The ARM and Power processors allow the relaxed outcome $r1=1$ and $r2=2$ in \reffig{fig:lb-sdep}, where there is nothing to enforce the order of the load and store in the second thread~\cite{DBLP:journals/pacmpl/PodkopaevLV19}. In \reffig{fig:lb-oota}, each thread reads and then writes the value read -- a data dependency from read to write. This dependency 
makes optimisation impossible, so the relaxed outcome $r1=r2=1$ is forbidden on every combination of compiler and target processor. C++ erroneously allows the outcome $r1=r2=1$, producing the value 1 ``out of thin air''.

Formally handling load buffering while avoiding out of thin air behaviours turns out to be an enormously complex task. Although several works \cite{DBLP:conf/popl/KangHLVD17,DBLP:conf/pldi/LeeCPCHLV20,ESOP2020-MRD,DBLP:journals/pacmpl/ChakrabortyV19,DBLP:journals/lmcs/JeffreyR19} have been dedicated to providing semantics for different variations of the load buffering example, many of these have also been shown to be inconsistent with expected behaviours under certain litmus tests \cite{ESOP2020-MRD,DBLP:journals/pacmpl/JagadeesanJR20}.
This work builds on top of the  recent \MRD (Modular Relaxed Dependencies) semantics by Paviotti et al~\cite{ESOP2020-MRD}. \MRD avoids thin air, and aims for compatibility with the existing ISO C and C++ standards. 

A key component of \MRD is the calculation of a {\em semantic dependency relation}, which describes when certain reorderings are disallowed. The program in \reffig{fig:lb-sdep} contains only the semantic dependency between lines 1 and~2, whereas the program in \reffig{fig:lb-oota} contains semantic dependencies between lines 1 and~2, and lines 3 and~4. The program in \reffig{fig:lb-sdep} may therefore execute line~4 before line 3, while the program in \reffig{fig:lb-oota} must execute both threads in order. 





\begin{figure}[t]
\noindent
\begin{minipage}[b]{0.44\columnwidth}
\centering
\begin{tabular}[t]{l||@{\ }l}
\multicolumn{2}{c}{\tt Init: x = y = r1 = r2 = 0}\\
\begin{minipage}[t]{0.44\columnwidth}
\tt Thread 1 \\
1: r1 := [x] \\
2: [y] := r1+1  \\
 \blue{\textrm{$\{r1 \in \{0, 1\}\}$}}
\end{minipage}
&
\begin{minipage}[t]{0.52\columnwidth}
\tt Thread 2 \\
3: r2 := [y] \\
4: [x] := 1  \\
\small \blue{\textrm{$\{r2 \in \{r1 + 1, 0\}\}$}}
\end{minipage}
\end{tabular}\\[2pt]
\blue{\{$r1 \neq 1 \lor r2 \neq 1\}$}
\vspace{-0.5em}
\caption{Load-buffering with a semantic dependency}
\label{fig:lb-sdep}
\end{minipage}
\hfill
\begin{minipage}[b]{0.4\columnwidth}
\centering
\begin{tabular}[t]{l||@{\ }l}
\multicolumn{2}{c}{\tt Init: x = y = r1 = r2 =  0}\\
\begin{minipage}[t]{0.42\columnwidth}
\tt Thread 1 \\
1: r1 := [x] \\
2: [y] := r1  \\
 \blue{\textrm{$\{r1 = 0\}$}}
\end{minipage}
&
\begin{minipage}[t]{0.52\columnwidth}
\tt Thread 2 \\
3: r2 := [x] \\
4: [x] := r2  \\
\small \blue{\textrm{$\{r2 = 0\}$}}
\end{minipage}
\end{tabular}\\[2pt]
\blue{\{$r1 = r2 = 0\}$}
\vspace{-0.5em}
\caption{Load-buffering with two dependencies}
\label{fig:lb-oota}
\end{minipage}
\end{figure}




\smallskip
\noindent {\it Contributions.}
Although precise, there is currently no direct mechanism for reasoning about programs under \MRD because \MRD is a denotational semantics defined over an event structure~\cite{DBLP:conf/ac/Winskel86}. This paper addresses this gap by developing an operational semantics for $\MRD$, which we then use as a basis for a deductive Owicki-Gries style verification framework. 
The key idea of our operational semantics is to take the semantic dependency relation as the only order in which programs must be executed, thereby allowing intra-thread reordering. This changes the fundamental meaning of sequential composition, allowing statements that occur ``later'' in the program to be executed early. 

Our semantics is also designed to take non multi-copy atomicity into account, whereby writes are not propagated to all threads at the same time and hence may appear take effect out-of-order~\cite{DBLP:journals/computer/AdveG96,AlglaveMT14}. Note that this phenomenon is distinct from the reordering of operations within a thread (described above), and it is possible to separate the two.
Here, we adapt the operational model of weak memory effects by Doherty et al \cite{DBLP:conf/ppopp/DohertyDWD19} so that it follows semantic dependency rather than the more restrictive thread order used in earlier works~\cite{DBLP:conf/ppopp/DohertyDWD19,DBLP:conf/ecoop/KaiserDDLV17,LahavV15,DBLP:conf/ecoop/DalvandiDDW19}.

Finally, we develop a logic capable of reasoning about program executions that exhibit both of the phenomena described above. Our logic makes use of the technique by Dalvandi et al~\cite{DBLP:conf/ecoop/DalvandiDDW19} of including assertions that enable reasoning about the ``views'' of each thread. Since we have concurrent programs, the logic we develop incorporates Owicki-Gries style reasoning for programs, in which assertions are shown to be both locally correct and globally stable (interference free). However, unlike earlier works~\cite{DBLP:conf/ecoop/DalvandiDDW19,LahavV15}, since we relax thread order, the standard approach to Hoare-style proof decomposition is not possible. 

\paragraph{\it Overview.}
This paper is organised as follows. We recap \MRD in \refsec{sec:mrd} and an operational semantics for RC11 in \refsec{sec:ppopp}. Then in \refsec{sec:opsem-mrd}, we present a combined semantics, where program order in RC11 is replaced by a more relaxed order defined by \MRD. A Hoare-like logic for reasoning about relaxed program execution together with Owicki-Gries-like rules for reasoning about interference is given in \refsec{sec:hoare-logic}. We present an example proof in \refsec{sec:verif-ex}.

%% file: background.tex
\section{\MRD and semantic dependencies}
\label{sec:mrd}
In this section, we review the \MRD semantics for a simple C-like while language. This provides a mechanism for defining (in a denotational manner) a relaxed order in which statements within a thread are executed, which precludes development of a program logic. 

\paragraph{Events and actions}
Weak memory literature uses a variety of terminology to refer to internal representations of changes to global memory, which complicates attempts to unify multiple models. In the following, we use \emph{events} to refer to the objects created and manipulated by \MRD, normally represented as integers. We use \emph{actions} to refer to objects of the form $\levR{i}{x}{v}$, $\levW{i}{x}{v}$, referring to a read or write of value $v$ at global location $x$ arising from line $i$ of the input program.

\paragraph{Program syntax.} 
\label{sec:syntax}
We assume {\em shared variables} $x, y, z, \ldots$ from a set $X$, {\em registers} $r_1, r_2, \ldots$ from a set $Reg$, and 
{\em register files} $\rho : Reg \to {\it Val}$ mapping registers to values from a set ${\it Val}$.  
    {\em Expressions} $e, e_1, e_2,\ldots$ are taken from a set $E$,
    whose syntax we do not specify,
    but which can be evaluated w.r.t. a register file
    using $\eval \in E \times (Reg \to {\it Val}) \to {\it Val}$. Thus,
    $\eval(e, \rho)$ is the value of $e$ given the
    register values in $\rho$. 
For basic commands, we have variable assignments
(or {\em stores})  $[y] := e$ 
and register assignments (or {\em loads}) $
    r := [x]$. In this paper, we assume that stores and loads are {\em relaxed} \cite{DBLP:conf/popl/BattyDG13,DBLP:conf/ppopp/DohertyDWD19,DBLP:conf/pldi/LahavVKHD17} unless they are explicitly specified to be a releasing store (denoted $[y] :=^R e$) or an acquiring load (denoted $
    r :=^A [x]$). 
Finally, we assume a simple language of {\em programs}. The  
only unconventional aspect of this language is that
we require each (variable or register) assignment be decorated with a unique {\em control label}. 
This allows a semantics which does not, in general, respect program order to refer to an individual statement without ambiguity.
The syntax of commands (for a single thread) is defined by the
following grammar, where $B$ is an expression that evaluates to a
boolean and $i$ is a control label. Note that since guards are expressions, they must not mention any shared variables --- any guard that relies on a shared memory variable must load its values into a local register prior to evaluating the guard. 

\smallskip \noindent \hfill 
\begin{tabular}[t]{r@{~}l}
  $\AComm$ ::= &  $i\!: \kwskip \mid i\!: [x] :=^{[R]} e \mid i\!:  r  :=^{[A]} [x] \mid i\!:  r := e
                $ \\
  $\Comm$ ::= & $\AComm \mid \Comm ; \Comm \mid \kwif~B\ \kwthen\ \Comm\ \kwelse\ \Comm 
   \mid \kwwhile\ B\ \kwdo\ \Comm$  
\end{tabular} \hfill{}\smallskip

\noindent We use $[x] :=^{[R]} e$ to denote that the releasing annotation $R$ is optional (similarly $r :=^{[A]} [x]$ for the acquiring annotation $A$). For simplicity, 
we focus attention on the core atomic features of C11 necessary to probe the thin-air problem, and omit more complex instructions such as CAS, fences, non-atomic accesses and SC accesses~\cite{DBLP:conf/popl/BattyOSSW11}.  

 We assume a top level parallel composition operator. Thus, a program is of the form $C_1 \| C_2 \| \ldots \| C_n$ where $C_i \in \Comm$. 
 We further assume each atomic command $\AComm$ in the program has a unique label across all threads. 

\paragraph{Denotational $\MRD$.}
For the purposes of this paper (i.e., the development of the operational semantics and associated program logic), the precise details of the $\MRD$ semantics~\cite{ESOP2020-MRD} are unimportant. The most important aspect that we use is the set of {\em semantic dependency} relations that it generates, which precisely characterise the order in which atomic statements are executed.

In $\MRD$, the {\em denotation} of a program $P$ is returned by the semantic interpretation function $\sem{P}$. This gives a {\em coherent event structure} of the form: 

$$
(L, S, \justif, \leq)
$$
where
\begin{itemize}
\item $L = (E, \sqsubseteq, \conflict, \mrdlab)$ is an {\em event structure}~\cite{DBLP:conf/ac/Winskel86} equipped with a labelling function $\mrdlab$ from events (represented internally as integer identifiers) to actions. The set $E$ contains all events in the structure, $a_1 \sqsubseteq a_2$ for $a_1, a_2 \in E$ iff $a_1$ is program ordered before $a_2$. Events in $\conflict$ cannot happen simultaneously, e.g., two reads of a variable returning different values are in conflict. 
\item $S = (A, \mrdlk, \mrdrf, \mrddep)$ is a set of {\em partial executions}, each of which represents a possible interaction of the program with the memory system. Each execution in $S$ is a tuple of relations representing {\em lock/unlock order}, a {\em reads-from relation} and {\em dependency order} between operations of {\em that} execution, and the set of events to be executed. An execution is \emph{complete} if every read in $A$ is linked to a write to the same location of the same value by an $\mrdrf$ edge.
\item $\justif$ is a {\em justification relation} used in the construction of the program's dependency relation. 
\item $\ppo$ is the {\em preserved program order} (c.f.,~\cite{AlglaveMT14}), a subset of the program order of $F$ that is respected by the memory model.
\end{itemize}

To develop our operational semantics, we only require some of these components. From $S$, we only require the dependency order. From the labelled event structure, we use the labelling function $\mrdlab$ to connect actions to events. From the coherent event structure, we use $\ppo$ to enforce ordering alongside dependency order.\footnote{$\MRD$ also defines a set of axioms that describes when a particular execution is consistent with a denotation. We do not discuss these in detail here, but they are used in the soundness and completeness proofs.}

\begin{example} 
The event structure in \reffig{fig:ces} represents the denotation of the program in \reffig{fig:lb-sdep}.
Note first that each store of a value into a register generates multiple events -- one for each possible value. This is because $\MRD$ cannot make assertions about which values it may or may not observe during interpretation of the structure, it can only be provided with global value range restrictions prior to running. Instead, each read must indicate a write whose value it is observing during the axiomatic checks, subject to various coherence restrictions. A read will only appear in a \emph{complete execution} if it can satisfy this requirement.

Events 2 and 4 are in \emph{conflict} (drawn as a red zigzag) as they represent different potential values being read at line 1. Likewise, events 6, 8, and 10 represent different potential values at line 3. A single execution can observe events 2 and 6, but not events 2 and 4.
Events 2 and 3 are in \emph{program order} (represented by the black arrow).
$\MRD$ generates a {\em dependency order} between events 2 and 3 and events 4 and~5 (represented by the yellow arrow). The read at line 1 is stored in register {\tt r1}, which is in turn accessed by the write in line 2, leading to a data dependency.
There is no semantic dependency between any events arising from lines 3 and~4, because there is no way for the value read at line 3 to influence the value written at line 4.
This means that lines 1 and 2 must be executed in order, but line 4 is free to execute before line 3. For example, consider the traces $H_1$ and $H_2$ below. 
\begin{align*}
H_1 \ \ & =\ \  {\rm Init}\ \ \levW{2}{y}{1} \ \ \levW{4}{x}{1}\ \ \levR{1}{x}{1}\ \ \levR{3}{y}{1} 
\\
H_2 \ \ & =\ \  {\rm Init}\ \ \levR{1}{x}{1}\ \ \levW{4}{x}{1}\ \ \levW{2}{y}{1} \ \ \levR{3}{y}{1}     
\end{align*}
For the program in \reffig{fig:lb-sdep}, the trace $H_1$ is {\em disallowed} since in $H_1$, the operations at lines $2$ and $1$ are swapped in a manner inconsistent with the dependency order. 
The trace $H_2$ is \emph{allowed}, since lines 3 and 4 are free to execute in any order.

\begin{figure}[t]
  \centering
  \scalebox{0.25}{\includegraphics{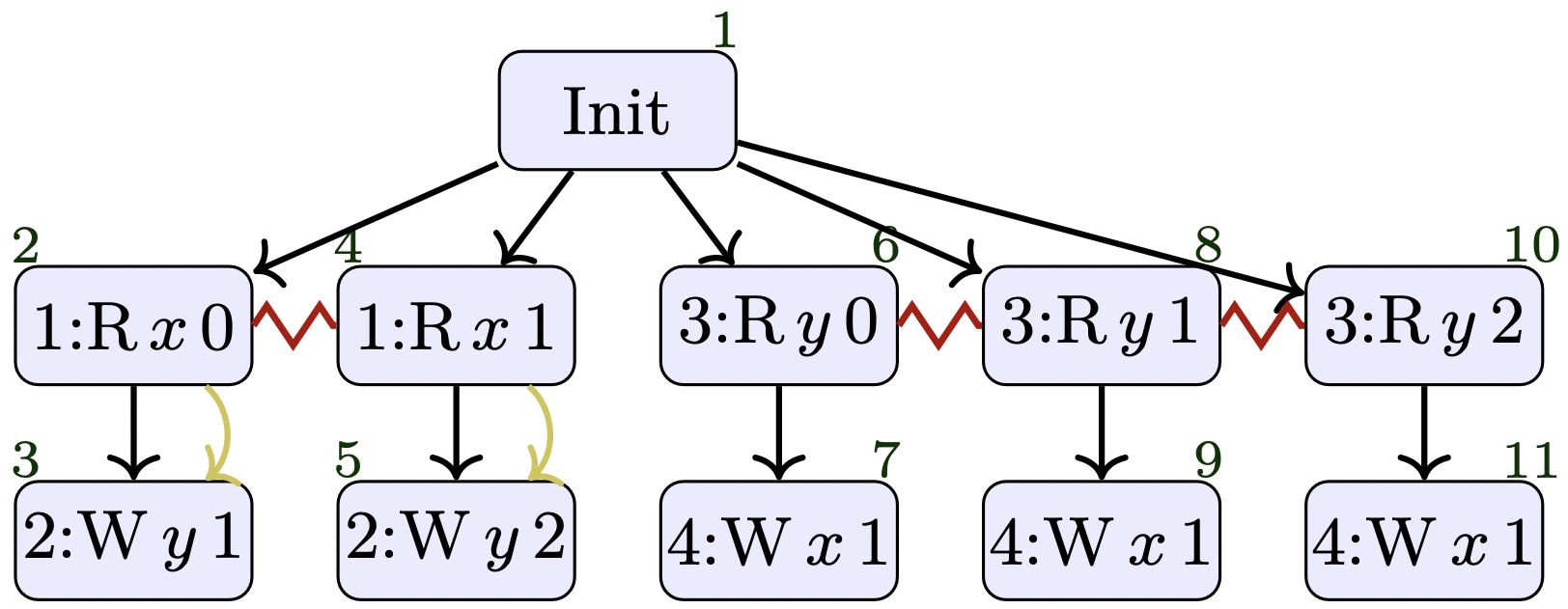}}
  
   \caption{Coherent event structure for the program in Fig.~\ref{fig:lb-sdep}}
    \label{fig:ces}
    \end{figure}

\end{example}





%% file: ppopp.tex
\section{Operational semantics with relaxed write propagation}
\label{sec:ppopp}



To reason about the relaxed propagation of writes, Doherty et al~\cite{DBLP:conf/ppopp/DohertyDWD19} build an operational semantics that is equivalent to a declarative (aka axiomatic) semantics~\cite{DBLP:conf/pldi/LahavVKHD17}, which defines consistency of {\em tagged action graphs}. The operational semantics allows the stepwise construction of such graphs, in contrast to a declarative approach which only considers complete executions that are either accepted or rejected by the axioms of the memory model. We do not discuss the declarative semantics here, focusing instead on the operational model, which we generalise in \refsec{sec:opsem-mrd}. We also note that like other prior works~\cite{DBLP:conf/ecoop/KaiserDDLV17,DBLP:conf/pldi/LahavVKHD17,LahavV15}, the existing operational semantics \cite{DBLP:conf/ppopp/DohertyDWD19} assumes program order within a thread is maintained.

Formally, {\em tagged actions} are triples $(g, a, t)$, where $g$ is a tag (uniquely identifying the action); $a$ is a read or write action (potentially annotated as a releasing write or acquiring read), and $t$ is a thread (corresponding to the thread that issued the action). We let $\TA$ be the set of all tagged actions, $\W, \W_R, \R, \R_A \subseteq \TA$ be the set of write, releasing write, read and acquiring reads, respectively. Note that $\W$ is the set of all writes, including those from $\W_R$ (similarly, $\R$).

A {\em tagged action graph} is a tuple $(D, \ltsb, \ltrf, \ltmo)$ where $D$ is a set of tagged actions, which may correspond to read or write actions, and $\ltsb$, $\ltrf$ and $\ltmo$ are relations over $D$. Here, $\ltsb$ is the {\em sequenced before} relation, where $b_1~\ltsb~ b_2$ iff $b_1$ and $b_2$ are tagged actions of the same thread and $b_1$ is executed before $b_2$.  $\ltrf \subseteq \W \times \R$ is the {\em reads-from} relation~\cite{AlglaveMT14} relating each write to the read that reads from that write. Finally, $\ltmo \subseteq \W \times \W$ denotes {\em modification order} (aka coherence order), which is the order in which writes occur in the system, and hence the order in which writes must be seen by all threads. Note that if $b_1~\ltmo~b_2$, then $b_1$ and $b_2$ must act on the same variable. Moreover, $\ltmo | x$ is a total order, where $\ltmo | x$ is the relation $\ltmo$ restricted to writes of variable $x$. 

We assume tagged action graphs are initialised with writes corresponding to the initialisation of the program, and relations $\ltsb$, $\ltrf$ and $\ltmo$ are initially empty. 

Following Lahav et al~\cite{DBLP:conf/pldi/LahavVKHD17}, the so-called repairing or restricted C11 model (which is the model in~\cite{DBLP:conf/ppopp/DohertyDWD19,DBLP:conf/ecoop/DalvandiDDW19,DBLP:conf/ecoop/KaiserDDLV17}) instantiates sequence-before order to the {\em program order} relation~\cite{AlglaveMT14}. This disallows statements within a thread from being executed out-of-order (although writes may be propagated to other threads in a relaxed manner). In \refsec{sec:opsem-mrd}, we present an alternative instantiation of $\ltsb$ using the semantic dependencies generated by $\MRD$ to enable out-of-order executions within a thread can be considered. 

To characterise the operational semantics, we must define three further relations: {\em happens before}, denoted $\lthb$ (which captures a notion of causality); {\em from read}, denoted $\ltfr$ (which relates each read to the write that overwrites the value read), and {\em extended coherence order}, denoted $\lteco$ (which fixes the order of writes and reads). Formally we have:
\begin{align*}
    \lthb = (\ltsb \cup (\ltrf \cap \W_R \times \R_A))^+
    \quad\ \ 
    \ltfr = (\ltrf^{-1} ; \ltmo) \setminus {\it Id}
    \quad\ \ 
    \lteco = (\ltrf \cup \ltmo \cup \ltfr)^+
\end{align*}
where {\it Id} is the identity relation, $;$ denotes relational composition, and ${}^+$ denotes transitive closure. Note that there is only a happens-before relation between a write-read pair related by $\ltrf$ if the write is releasing and the read is acquiring. 



\begin{figure}[t]
  \centering
  \small
  $ \inference[{\sc Read}] {b = (g, a, t) \qquad g \notin tags(D) \qquad a \in \{\levR{i}{x}{n},
    \levRA{i}{x}{n} \}    \\
   \sigma = (D, \ltsb,\ltrf,\ltmo)\qquad w \in \OW_\sigma(t) \qquad  \loc(w) = x \qquad \wrval(w)=n} {(D,\ltsb, \ltrf,\ltmo) \strans{\ b\ } (D \cup \{b\}, \ltsb +_D b, \ltrf \cup
    \{(w,b)\},\ltmo)}$
  \bigskip

  $ \inference[{\sc Write}] {
    b = (g, a, t) \qquad g \notin tags(D) \qquad a \in \{ \levW{i}{x}{n},
    \levWR{i}{x}{n}\} \\ \sigma = (D, \ltsb,\ltrf,\ltmo) \qquad w \in \OW_\sigma(t) 
    \qquad \loc(w) = x 
  }{(D,\ltsb, \ltrf,\ltmo) \strans{\ b\ } (D \cup \{b\},\ltsb +_D b, \ltrf,\ltmo[w,b])}$


  \caption{Memory semantics}
  \label{fig:c11-opsem}
\end{figure}

With these basic relations in place, we are now in a position to define the transition relation governing the operational rules for read and write actions. The rules themselves are given in \reffig{fig:c11-opsem}. Assuming $\sigma$ denotes the set of all tagged action graphs, each transition is a relation $
\mathord{\strans{\ \ }} \subseteq \Sigma \times \TA \times \Sigma$. We write
$\sigma \strans{\ b\ } \sigma'$ to denote
$(\sigma, b, \sigma') \in \mathord{\strans{\ \ }}$.

To accommodate relaxed propagation of writes, the semantics allows different threads to have different views of the system, formalised by a set of {\em observable writes}. These are in turn defined in terms of a set of {\em encountered writes}, denoted $\observedWrites_\sigma(t)$, which are the writes that thread $t$ is
aware of (either directly or indirectly) in state $\sigma$:
\begin{align*}
  \observedWrites_\sigma(t) = \{w  \in \W \cap D_\sigma \mid \exists b \in D_\sigma.\
  \begin{array}[t]{@{}l@{}}
    \tid(b) = t \wedge {} 
    (w,b) \in \lteco_\sigma^?; \lthb_\sigma^? \}
  \end{array}
\end{align*}
Here $R^?$ is the reflexive closure of relation $R$ and $\tid$ returns the thread identifier of the given tagged action.  Thus, for each
$w \in \observedWrites_\sigma(t)$, there must exist a tagged action $b$ of
thread $t$ such that $w$ is either $\lteco$-, $\lthb$- or
$\lteco ; \lthb$-prior to $b$. 
From these we determine the \emph{observable writes}, which are the
writes that thread $t$ can observe in its next read. These are defined
as:
\begin{align*}
  \OW_\sigma(t) = \{w  \in \W \cap D_\sigma  \mid \forall w' \in \observedWrites_\sigma(t).\ (w,w') \notin \ltmo_\sigma\}
\end{align*}
Observable writes are writes that are not succeeded by any encountered write in
modification order, i.e., the thread has not seen another write
overwriting the value being read.

We now describe each of the rules in \reffig{fig:c11-opsem}. 
Relations $\ltrf$ and $\ltmo$ are updated according to the
write actions in $D$ that are observable to the thread $t$ executing the
given action. A read action $b$ may read from any write $w \in \OW_\sigma(t)$. In the post state, we obtain a new tagged action set $D \cup \{b\}$, sequenced-before relation $\ltsb +_D b$ (defined below) and reads-from relation $\ltrf \cup \{(w, b)\}$. Relation $\ltmo$ is unmodified. Formally, $\ltsb +_D b$ introduces $b$ at the end of $\ltsb$ for thread $t$:
\begin{align*} 
  \ltsb +_D b
  & =
\begin{array}[c]{@{}l@{}}
      \ \ltsb \cup  (\{b' \in D  \mid \tid(b') \in \{\tid(b), 0\} \}
      \times \{b\} )
    \end{array}
\end{align*} 
Note that we assume that initialisation is carried out by a unique thread with id $0$, and that we require the set $D$ as an input to cope with initialisation where $\ltsb$ may be empty for the given thread. 

The write rule is similar except that it leaves $\ltrf$ unchanged and updates $\ltmo$ to $\ltmo[w, b]$. 
Given that $R[x]$ is the relational image of $x$ in $R$, we define
$R_{\Downarrow x} = \{x\} \cup R^{-1}[x]$ 
to be the set of all elements in $R$ that relate to $x$ (inclusive).
The insertion of a tagged write action $b$ directly after a $w$ in
$\ltmo$ is given by
\[ \textstyle\ltmo[w,b] =
  \begin{array}[t]{@{}l@{}}
    \ltmo  \cup 
    (\ltmo_{\Downarrow w} \times \{b\}) \cup (\{b\} \times
    \ltmo[w]) 
  \end{array}
\]
Thus, $\ltmo[w,b]$ effectively introduces $b$ immediately after the $w$ in $\ltmo$.

The final component of the operational semantics is a set of rules that link the program syntax in \refsec{sec:syntax} with the memory semantics. In prior work, this was through a set of operational rules that generate the actions associated with each atomic statement, combined with the rules in \reffig{fig:c11-opsem} to formalise the evolution of states. A sample of these rules for reads and writes from memory are given in \reffig{fig:comm-sem-sample}. 
These are then lifted to the level of programs using the rules {\sc Prog} and {\sc Full}, where the transition rule $\ltsArrow{b}$ combines the thread-local semantics $\whilestep{a}_t$ and global-state semantics $\strans{\ b\ }$. Note that we model parallel composition as functions from thread identifiers to commands, thus $P[t
    := C]$ represents the program, where the command for thread $t$ is updated to $C$.   

\begin{figure*}[!t]
\begin{center}

  $ \inference{n = eval(e, \rffun(t)) \quad a = \levWRO{i}{x}{n}}{(i\!: [x] :=^{[R]} e, \rffun)
    \whilestep{a}_t (\kwskip, \rffun) } $
  \qquad
  $ \inference{ a = \levRAO{i}{x}{n} \quad \rho' = \rffun(t)[r := n]}{(i\!: r :=^{[A]} [x], \rho)
    \whilestep{a}_t (\kwskip, \rffun[t := \rho']) } $
  \bigskip


$\text{\sc Prog}
  \inference{(P(t), \rffun) \whilestep{a}_t (C,\rffun')}
  {(P, \rffun) \whilestep{a}_t (P[t
    := C], \rffun')}$
\qquad
$\inference[\sc Full]{b = (g, a, t) \vspace{-0.25em} \\ 
(P, \rffun)
   \whilestep{a}_t (P', \rffun\,') \qquad  \sigma \strans{\ b\ } \sigma'
  } 
  {(P, (\sigma, \rffun)) \ltsArrow{b}
  (P', (\sigma', \rffun\,'))}
$
\vspace{-3mm}
\end{center}
  \caption{Interpreted operational semantics of programs (sample)}
  \label{fig:comm-sem-sample}
\end{figure*}

The rules in \cite{DBLP:conf/ppopp/DohertyDWD19} generate actions corresponding to program syntax (i.e., $\whilestep{a}_t$) in {\em program order}.\footnote{We present the rules from \cite{DBLP:conf/ppopp/DohertyDWD19} in Appendix~\ref{sec:full-op-sem}.} In this paper, we follow a different approach --- the actions will be generated by (and executed in) the order defined by \MRD. This therefore allows both reordering of program statements and relaxed write propagation.

%% file: opsem.tex
\section{Operational Semantics over $\MRD$}
\label{sec:opsem-mrd}

\newcommand{\labfun}{\lambda_{\mu}}




Recall that the $\MRD$ semantic interpretation function $\sem{P}$ outputs a coherent event structure $\mu = (L, S, \vdash, \ppo)$. In this section, we develop an operational semantics for traversing $\mu$, while generating state configurations that model relaxed write propagation. In essence, this generalises the operational semantics in \refsec{sec:ppopp} so that threads are executed out-of-order, as allowed by $\MRD$.



\subsection{Program Futures}

\label{sec:f-trans-rel}
Defining our operational semantics over raw syntax would force us to evaluate statements in program order. Therefore we convert this syntax into a set of atomic statements. 
We call this the \emph{atomic set} of $C$, written $\pset{C}$. $\MRD$ evaluates while loops using step-indexing, treating them as finite unwindings of $\kwif$-$\kwthen$-$\kwelse$ commands. For these, we generate a fresh unique label for each iteration of the while loop for each of the atomic commands within the loop body.
Recall that the parallel composition of commands is modelled by a function from thread identifiers to (sequential) commands. The atomic set of a program $P$ is therefore $\lambda t.\ \pset{P(t)}$.  

To retain the ordering recognised during program execution, we introduce \emph{futures}, which are sets of $\MRD$ events partially ordered by the semantic dependency and preserved program order relations. Essentially, instead of taking our operational steps in program order over the syntax, we can nondeterministically execute any statement which our futures tell us we have executed all necessary predecessors of.

For an execution $S = (A, \mrdlk, \mrdrf, \mrddep)$ of an event structure $\mu$, we can construct an initial future $f = (A, \preceq)$, where  $\mathop{\preceq} = \mathop{\mrddep} \cup \mathop{\leq_{|A}}$ and $\leq_{|A}$ is the preserved program order $\ppo$ (see \refsec{sec:mrd}) restricted to events in $A$.
We say that an action $a$ is {\em available} in a future $(K, \preceq)$ iff
there exists some event $g$ with label $a$ such that
 $g \in K$ and $g$ is minimal in $\preceq$, i.e., for all events $g' \in K$, $g' \not\preceq g$.
If $a$ is available in $f$, then 
the {\em future of $a$ in $f$}, denoted $a \future f$, is the future
\begin{align*}
a \future f & \iff (K', \preceq_{|K'}) & \text{where $K' = K \setminus \{g \mid \labfun(g)  = a\}$}    
\end{align*}
We lift this to a set of futures $F$ to describe the {\em candidate
futures of $a$ in $F$}, denoted $a \future F$:
\begin{align*}
a \future F = \{a \future f \mid f \in F \wedge \text{$a$ available in $f$}\}
\end{align*}
Note that $a$ is {\em enabled} in $F$ iff $a \future F \neq \emptyset$.  

Essentially, $a \future F$ consumes an event $g$ with label $a$ from each of the futures in $F$ provided $a$ is available, discarding all futures in which $a$ is not available. 
Intuitively, if an event is minimal in a program future then it may be executed immediately. If it is not minimal, its predecessors must be executed first.

\subsection{Future-based Transition Relation}
Our operational semantics is defined by the transition relation in {\sc Future-Step} given below, which generalises {\sc Full} in \reffig{fig:comm-sem-sample}. The transition relation is defined over
$(\pset{Q}, (\sigma, \rffun), F)$, where $\pset{Q}$ is the atomic set corresponding to the program text, $(\sigma, \rffun)$ is a configuration and $F$ is a set of {\em futures}. The rule generalises {\sc Full} in the obvious way, i.e., by evolving the configuration as allowed
by $\strans{(g, a, t)}$ and $\whilestep{a}_t$ and consuming an available action $a$ in $F$. 
Below, we use $\uplus$ to denote disjoint union and $f[k := v]$ to denote function updates where $f(k)$ is updated to $v$.
$$\inference[\sc Future-Step] {
     \pset{Q}(t) = \pset{C} \uplus \{i\!:\ s\}
     \qquad
     a \future F \neq \emptyset \vspace{-0.25em} \\
     (i: s, \rffun) \whilestep{a}_t (\kwskip, \rffun')
     \qquad      \sigma \strans{(g, a, t)} \sigma'
     }{
     (\pset{Q}, (\sigma, \rffun), F)\ {\ltsArrow{\mu}}\ (\pset{Q}[t := \pset{C}], (\sigma', \rffun'), a \future F)}
$$
Once a minimal action is chosen in a step of the operational semantics, it is checked for consistency and added to the set of executed events. The set of futures is pruned to remove the chosen event, and to exclude any futures that are incompatible with the chosen event. The operational semantics continues until it has consumed all of the futures.

The following theorem establishes equivalence of our operational semantics and the $\MRD$ denotational semantics. A proof is provided in the appendix.
\newcommand{\bij}{\mathcal{E}}
\begin{theorem}[Soundness and Completeness]
Every execution generated by the $\MRD$ model can be generated by the operational semantics, and every final state generated by the operational semantics corresponds to a complete execution of the $\MRD$ semantics.
\end{theorem}


\subsection{Example}
Recall the program in \reffig{fig:lb-sdep} and its event structure representation in \reffig{fig:ces}. Let $\Delta_S = \{(x, x) \mid x \in S\}$ be the diagonal of set $S$. We first derive our set of futures from this structure:

\begin{center}
\begin{tabular}{l@{\qquad}l@{\qquad}l}
    $\{2 \prec 3, 6, 7\}$ & $\{2 \prec 3, 8, 9\}$ & $\{2 \prec 3, 10, 11\}$\\    
    $\{4 \prec 5, 6, 7\}$ & $\{4 \prec 5, 8, 9\}$ & $\{4 \prec 5, 10, 11\}$
\end{tabular}
\end{center}
where $\{2 \prec 3, 6, 7\}$ represents the future $(\{2,3,6,7\}, 2 \prec 3 \cup \Delta_{\{2,3,6,7\}})$.
The atomic set of the program is 
$$\pset{P} = \left\{
\begin{array}[c]{l}
1\mapsto \{{\tt 1: r1 := [x]}, {\tt 2: [y] := r1 + 1}\}, \\
2 \mapsto \{{\tt 3: r2 := [y]}, {\tt 4: [x] := 1}\}
\end{array}
\right\}$$
The initial configuration is $(\sigma_0, \gamma_0)$, where 
$\sigma_0 = (\{(0_x, \levW{0}{x}{0}, 0), (0_y, \levW{0}{y}{0}, 0)\}, $ $  \emptyset, \emptyset, \emptyset)$ and $\gamma_0 = \{1 \mapsto \{r1 \mapsto 0\}, 2 \mapsto \{r2 \mapsto 0\}\}$, assuming the initialising thread has identifier~0.


To find out which events we can execute, we check our futures set.
We cannot execute events 3 or 5, which both have pre-requisite events that have not yet been executed. These events are the only available events generated by line 2, so we cannot attempt to execute line 2. We can, however, execute any other line.
Suppose we execute line 4, which corresponds to ${\tt 4: [x] := 1}$ in $\pset{P}(2)$. 
To use the transition relation $\ltsArrow{\mu}$, we need to do the following: (1) determine the corresponding action, $a$ and new thread-local state using $\whilestep{a}_2$; (2) generate a tagged action $b = (g, a, 2)$ for a fresh tag $g$ and a new global state using $\strans{b}$; and (3) check that $a$ is available in the current set of futures. 

For (1), we can only create one action $a = \levW{4}{x}{1}$ and the local state is unchanged. For (2), we generate a new global state using the {\sc Write} rule in \reffig{fig:c11-opsem} (full details elided). For (3), we take our candidate futures $a \future F$ by examining which $\MRD$ events have the label $\levW{4}{x}{1}$ --- in this case events 7, 9, and 11. All futures can execute one of these events so the new future set contains all futures in $F$, each minus the set $\{7, 9, 11\}$.







\input{cond_op_sem}



%% file: cond_op_sem.tex
\newcommand{\dumact}{\textsc{cond}}
\newcommand{\governs}{\textsc{governs}}
\newcommand{\testa}{\textsc{test}}



%% file: hoare-logic.tex
\section{Hoare Logic and Owicki-Gries Reasoning}
\label{sec:hoare-logic}

Recall that the standard Owicki-Gries methodology~\cite{DBLP:journals/acta/OwickiG76} decomposes proofs of parallel programs into two cases:
\begin{itemize}
\item {\em local correctness} conditions, which define correctness of an assertion with respect to an individual thread, and  
\item {\em non-interference} conditions, which ensures stability of an assertion under the execution of statements in other threads.
\end{itemize}
The standard Owicki-Gries methodology has been  shown to be applicable to a weak memory setting with relaxed write propagation (but without relaxed program order)~\cite{DBLP:conf/ecoop/DalvandiDDW19}. Note that an alternative characterisation (also in a model without relaxed program order) has been given in \cite{LahavV15}.\footnote{This characterisation uses standard assertions but assumes a non-standard interpretation of Hoare-triples and introduces a stronger interference freedom check. In fact, for the model in \cite{LahavV15}, the introduction of auxiliary variables is unsound.}
Unfortunately, the unrestricted  C11 semantics captured by $\MRD$ relaxes program order and hence sequential composition in order to allow behaviours such as those in \reffig{fig:lb-sdep}, which requires a fundamental shift in Hoare-style proof decomposition. We show that the modular Owicki-Gries rules, for reasoning about concurrent threads remain unchanged.

Like Dalvandi et al~\cite{DBLP:conf/ecoop/DalvandiDDW19}, we assume assertions are predicates over state configurations. 
In the current paper, the operational semantics is dictated by the set of futures generated by $\MRD$, thus we require two modifications to the classical meaning. First, like prior work \cite{DBLP:conf/ecoop/DalvandiDDW19}, we assume predicates are over state configurations, which include (local) register files and (shared) event graphs. This takes into account relaxed write propagation. Second, we introduce Hoare-triples with futures generated by $\MRD$, which takes into account relaxed program execution. 

In the development below, we refer to the {\em preprocessed form} of a program $P$ given by $\preprocess(P) = (\mu, \pset{P})$, where $\mu$ is the coherent event structure $\sem{P}$ and  $\pset{P}$ the set normal form of $P$. We let $F_\mu$ be the initial set of futures corresponding to $\mu$.   

\begin{definition}[Hoare triple]
\label{def:hoare-triple-1}
Suppose $X$ and $Y$ are predicates over state configurations,  $P$ is a concurrent program, and $\preprocess(P) = (\mu, \pset{P})$. The semantics of a Hoare triple is given by $\{X\}\ {\it Init} ; P\ \{Y\}$, where 
\begin{align*}
 \{X\}\ {\it Init} ; P\ \{Y\}
 \ \ \sdef \ \  & 
    \begin{array}[t]{@{}l@{}}
         ({\it Init} \imp  X)  \wedge {} \\
      (\forall \sconfigc, \sconfigc'.\ 
      X(\sconfigc)  \wedge {}
      ((\pset{P}, \sconfigc, F_\mu) \ltsArrow{\mu}^* (\emptyset, \sconfigc', \emptyset))
      \ \imp \  Y(\sconfigc'))
    \end{array}
\end{align*}
\end{definition}
 That is, $\{X\}\ P\ \{Y\}$ holds iff for every pair of state configurations $\sconfigc$, $\sconfigc'$, assuming $X(\sconfigc)$ holds and we execute  $\pset{P}$ with respect to the future $F_\mu$ until $\pset{P}$ terminates in $\sconfigc'$, we have $Y(\sconfigc')$. 



\label{sec:proof-method}

Although \refdef{def:hoare-triple-1} provides  meaning for a Hoare triple, we still require a method for decomposing the proof outline. To this end, we introduce the concept of a {\em future predicate}, which is a predicate parameterised by both futures and configuration states. 
To make use of future predicates, we introduce a notion of a Hoare triple for programs in set normal form.

For the definitions below, assume $P$ is a program and $\mrd(P) = (\mu, \pset{P})$.
We say an atomic set $\pset{Q}$ is a {\em sub-program} of an atomic set $\pset{P}$ iff for all $t$, $\pset{Q}(t) \subseteq \pset{P}(t)$. 
We say a set of futures $F'$ is a {\em sub-future} of set of futures $F$ iff for each $f' \in F'$ there exists an $f \in F$ such that $f'$ is an up-closed subset of $f$. For example, if $F = \{\{1 \prec 2, 3, 4\}, \{1 \prec 2, 3\}\}$, then $\{\{2, 4\}, \{1 \prec 2\}\}$ is a sub-future of $F$,  but $\{\{1, 3\}\}$ is not. We say the sub-program $\pset{Q}$ corresponds to a sub-future $F$ iff  $labels(\pset{Q}) = labels(\{\mrdlab_\mu [\pi_1 f] \mid f \in F\})$, where we assume $labels$ returns the set of all labels of its argument, $\pi_1$ is the project of the first component of the given argument, and $R[S]$ is the relational image of set $S$ over relation $R$. 








\begin{definition}[Hoare triple (single step)]
\label{def:hoare-triple}
Suppose $I$ and $I'$ are future predicates, $P$ is a program and $\mrd(P) = (\mu, \pset{P})$. If $G$ is a sub-future of $F_\mu$, corresponding to a sub-program $\pset{Q}$ of $\pset{P}$, we define 
\begin{align*}
 & \{I\}_G\ \pset{Q}\ \{I'\} \\ 
 \sdef\ & 
     \forall \sconfigc, \sconfigc', G'.\ 
      I(G)(\sconfigc)  \wedge 
      ((\pset{Q}, \sconfigc, G) \ltsArrow{\mu} (\pset{Q}', \sconfigc', G'))
      \ \imp \ I'(G')(\sconfigc') 
\end{align*}
\end{definition}
We say $I$ is {\em future stable} for $(F, \pset{P})$ iff for all sub-futures $G$ of $F$ with corresponding sub-programs $\pset{Q}$ of $\pset{P}$, we have $\{I\}_G\ \pset{Q}\ \{I\}$.

\begin{lemma}[Invariant]
Suppose $X$ and $Y$ are configuration-state predicates, $P$ is a program and $\mrd(P) = (\mu, \pset{P})$. If $X \imp I(F_\mu)$, $I(\emptyset) \imp Y$ and $I$ is future stable for $(F_\mu, \pset{P})$, then 
$\{X\} {\it Init}; P \{Y\}$ provided ${\it Init} \imp X$.
\end{lemma}

By construction, the sets of futures corresponding to different threads are disjoint, i.e., for each future $f \in F_\mu$, we have that $f = \bigcup_t f_{| t}$, where $f_{|t}$ denotes the future $f$ restricted to events of thread $t$. The only possible inter-thread dependency in $\MRD$ is via the reads from relation~\cite{ESOP2020-MRD}, which does not contribute to the set of futures. This observation leads to a technique for an Owicki-Gries-like modular proof technique for decomposing the monolithic invariant $I$ into an invariant per thread. 

We define $F_{|t} = \{f_{| t} \mid f \in F\}$. Then, we obtain the following lemma for decomposing invariants in the same way as Owicki and Gries. 

\begin{lemma}[Owicki-Gries]
\label{lem:og}
For each thread $t$, let $I_t$ be a future predicate corresponding to $t$. If ${\it Init} \imp X$, $X \imp \forall t.\ I_t({F_\mu}_{| t})$, and $\forall t.\ I_t(\emptyset) \imp Y$, then $\{X\} P \{Y\}$ holds provided both of the following hold.  
\begin{enumerate}
    \item For all threads $t$, $I_t$ is future stable for $({F_\mu}_{| t}, \pset{P})$. \hfill (local correctness)

    \item For all threads $t_1$, $t_2$ such that $t_1 \neq t_2$, sub-futures $F_1$ of ${F_\mu}_{| t_1}$, and $F_2$ of ${F_\mu}_{| t_2}$, if $\pset{Q}$ corresponds to $F_2$, then  we have\footnote{Technically speaking, each instance of $I_{t_1}(F_1)$ in the Hoare-triple is a function $\lambda x.\ I_{t_1}(F_1)$.}
    $\{I_{t_1}(F_1) \wedge I_{t_2}\}_{F_2}\ \pset{Q}\ \{I_{t_1}(F_1)\} $.
    
    \hfill (global  correctness)    
\end{enumerate} 
\end{lemma}
Thus, we establish $\{X\} P \{Y\}$ through a series of smaller proof obligations. We require that (1) the initialisation of the program guarantees $X$, (2) whenever $X$ holds then for each thread $t$, $I_t$ holds for the initial future ${F_\mu}_{ | t}$, (3) if $I_t(\emptyset)$ holds for all $t$, then the post-condition $Y$ holds, (4) each $I_t$ is maintained by the execution of each thread, and (5) $I_t$ at each sub-future of $t$ is stable with respect to  steps of another thread. 



%% file: example.tex
\section{A Verification Example}
\label{sec:verif-ex}
With the verification framework now in place, we present a correctness proof for the program in \reffig{fig:lb-sdep}. The program and invariant are given in \reffig{fig:lb-sdep-po}. First, in \refsec{sec:assertions}, we present assertions for reasoning about state-configurations.

\subsection{View-based assertions}
\label{sec:assertions}
As we can see from \reffig{fig:comm-sem-sample}, the states that we use are {\em configurations}, which are pairs of the form $(\sigma, \rffun)$, where $\sigma$ is an tagged action graph representing the shared state and $\rffun$ is mapping from threads to register files representing the local state. Like prior work~\cite{DBLP:conf/ecoop/DalvandiDDW19}, we use assertions that describe the {\em views} of each thread, recalling that due to relaxed write propagation, the views of each thread may be different. Note that the formalisation of a state in prior work is a time-stamp based semantics~\cite{DBLP:conf/ecoop/DalvandiDDW19}. Nevertheless, the principles for defining thread-view assertions also carry over to our setting of tagged action graphs.  

In this paper, the programs we consider are relatively simple, and hence, we only use two types of view assertions: {\em synchronised view}, denoted $[x = v]_t$, which holds iff both thread $t$ observes the last write to $x$ and this write updates the value of $x$ to $v$, and {\em possible view}, denoted $[ x \approx v]_t$, which holds iff $t$ may observe a write to $x$ with value $v$. Formally, we define 
\begin{align*}
    [x = v]_t (\sigma, \gamma) & \sdef \exists w.\ \OW_\sigma(t)_{|x} = \{w\} \wedge \wrval(w) = v \\
    [x \approx v]_t (\sigma, \gamma) & \sdef \exists w \in \OW_\sigma(t)_{|x}.\  \wrval(w) = v 
\end{align*}
where $\OW_\sigma(t) | x$ denotes the writes in $\OW_\sigma(t)$ restricted to the variable $x$. Recall that, by definition, for any state $\sigma$ generated by the operational semantics, the last write to each variable in $\ltmo$ order is observable to every thread. Examples of these assertions in the context of an Owicki-Gries-style proof outline is given in \reffig{fig:lb-sdep-po}. 

\subsection{Example proof}
To reduce the domain of our future predicate, we collapse the event-based futures used by the operational semantics into sets of label-based futures. An event future $F_E$ can be converted into a label future $F_L$ by applying the labelling function to all events in $F_E$. This makes the futures $\{3\}$ and $\{5\}$ equivalent, as both are instances of $\{\levW{2}{y}{2}\}$, thus $I(\{3\}) = I(\{5\})$. This isn't always a valid step, as some events which share labels may not be related by $\preceq$ in the same way. Thus, the technique can only be used if the label-based representation describes exactly the futures generated by $\MRD$, which is the case for our example. 

We describe our initial set of label futures as $\{\{1_u < 2, 3_v, 4\} \mid u \in \{0, 1\} \wedge v \in \{0,1,2\}\}$, using the notation $\{i_v\}$ to refer to an action with the line number $i$ with a read that returns the value $v$.
We verify that this is a valid step: all futures generated by $\MRD$ are described by these labels, and they do not describe any potential futures not generated by $\MRD$.
We partition this into $F = \{\{1_u < 2\} \mid u \in \{0,1\}\}$ and $G = \{\{3_v, 4\} \mid v \in \{0,1,2\}\}$ representing the future sets of threads $1$ and $2$, respectively. We let $F_1  = \{\{2\}\}$ be the future set after executing line $1$, $G_3 = \{\{4\}\}$ be the future set after executing line $3$ (reading some value for $y$), and $G_4 = \{\{3_v\} \mid v \in \{0,1, 2\}\}$ be the future set after executing line $4$. 

\begin{figure}[t]
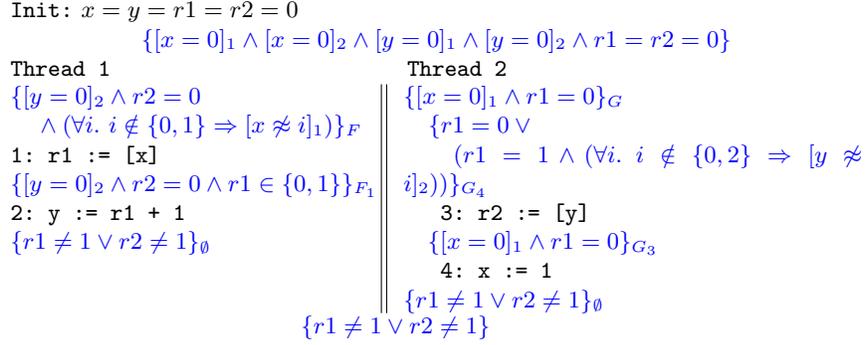


\centering

\noindent
\begin{tabular}[t]{r||@{\ \ }l}
\multicolumn{2}{l}{{\tt Init:} $x = y = r1 = r2 = 0$}\\
\multicolumn{2}{c}{\blue{\textrm{$\{[x = 0]_1 \wedge [x = 0]_2 \wedge [y = 0]_1 \wedge [y = 0]_2 \wedge r1 = r2 = 0 \}$}}}\\
\multicolumn{1}{l}{\tt Thread 1} & 
\multicolumn{1}{l}{\tt Thread 2} \\
\begin{minipage}[t]{0.4\columnwidth}\tt
 \vspace{-0.75em}
 \blue{\textrm{$\{[y = 0]_2 \wedge r2 = 0$}} \\
 \blue{\textrm{$\ \ \ {} \wedge (\forall i.\ i \notin \{0,1\} \imp [x \not\approx i]_1)\}_F$}} \\
1: r1 := [x] \\
 \blue{\textrm{$\{[y = 0]_2 \wedge r2 = 0 \wedge r1 \in \{0,1\}\}_{F_1}$}} \\
2: y := r1 + 1  \\
\blue{$\{r1 \neq 1 \lor r2 \neq 1\}_\emptyset $}
\end{minipage}
&
\begin{minipage}[t]{0.5\columnwidth} \tt 
 \vspace{-0.75em}
 \blue{\textrm{$\{[x = 0]_1 \wedge  r1= 0\}_{G}$}} \\
 \blue{\quad \textrm{${} \{r1 = 0 \lor {}$}} \\
 \blue{\quad \quad  \textrm{$(r1 = 1 \wedge (\forall i.\ i \notin \{0,2\} \imp [y \not\approx i]_2))\}_{G_4}$}} \\
 \phantom{\textrm{\quad}} 3: r2 := [y] \\
 \blue{\textrm{$\quad \{[x = 0]_1  \wedge r1 = 0\}_{G_3}$}} \\
\phantom{\textrm{\quad}} 4: x := 1  \\
\blue{$\{r1 \neq 1 \lor r2 \neq 1\}_\emptyset$}
\end{minipage}\\
\multicolumn{2}{c}{\blue{\{$r1 \neq 1 \lor r2 \neq 1$\} \qquad \quad }} 
\end{tabular}
\caption{Proof outline for load buffering with semantic dependencies}
\label{fig:lb-sdep-po}
\end{figure}
Our future predicate $I$ must output a configuration predicate for every sub-future of the initial set. Our partitioning fully describes these sub-futures, so we now attach our assertions to these futures.


To simplify the visualisation, we interleave the future predicate components with the program to provide Hoare-style pre/post-assertions, and place the assertion above a line of code if that line of code is contained in the applied future. We use indentation to denote that an assertion applies to more than one future. In this example $G$ contains both lines $3$ and $4$, hence both lines are indented w.r.t. the first assertion in thread $2$. 
We apply \reflem{lem:og}, and in the discussion below, we describe the local and global correctness checks.

For local correctness in thread $1$, we must establish that
$\{I\}_F \{1, 2\} \{I\}_{\emptyset}$.
By the definition of the future $F$, line~1 must be executed before line 2. 
This means we only need to verify that 
$\{I\}_F \{1\} \{I\}_{F_1}$ and
$\{I\}_{F_1} \{2\} \{I\}_{\emptyset}$, which is identical to a standard Hoare logic proof and relatively uninteresting.

In thread $2$, no order is imposed between lines $3$ and~$4$. 
This means to establish local correctness we must check that:
\begin{itemize}
    \item $\{I\}_G\;\; \texttt{3: r2 := y}\; \{I\}_{G_3}$ 
    \item $\{I\}_G\;\; \texttt{4: x := 1}\; \{I\}_{G_4}$ 
    \item $\{I\}_{G_3}\; \texttt{4: x := 1}\; \{I\}_{\emptyset}$
    \item $\{I\}_{G_4}\; \texttt{3: r2 := y}\; \{I\}_{\emptyset}$
\end{itemize}

The first three are trivial: line~3 modifies neither $x$ nor $r1$ and line 4 does not modify $r1$.
For the final step, the first disjunct of $\{I\}_{G_4}$ is the same as the first disjunct of $\{I\}_\emptyset$, and the second disjunct ensures that we cannot observe $r2 = 1$ after executing line 3.


For global correctness, we must check that every assertion in thread 1 continues to hold after every line of thread 2, and vice versa. These checks are also straightforward, so omit a detailed discussion. The only noteworthy aspect is that for line $3$ (and similarly, line $4$), our precondition is the conjunction $\{I\}_{G_4} \wedge \{I\}_G$, as both $G$ and $G_3$ are subfutures corresponding to line 3.

The proof described above has actually been encoded and checked using our existing Isabelle/HOL development~\cite{DBLP:conf/ecoop/DalvandiDDW19,DBLP:journals/corr/abs-2004-02983} by manually encoding the re-ordering in thread~$2$. We aim to develop full mechanisation support as future work.

%% file: conclusions.tex
\section{Related work}
\label{sec:related}
\paragraph{Work based on assumptions unsound in C++.} 
Most logics for weak memory are based on simplifying assumptions that exclude thin air reads by ensuring program order is respected, even when no semantic dependency exists~\cite{DBLP:conf/ecoop/KaiserDDLV17,DBLP:conf/ecoop/DalvandiDDW19,DBLP:conf/ppopp/DohertyDWD19,DBLP:conf/esop/DokoV17,LahavV15}. These assumptions incorrectly exclude the relaxed outcome of load buffering tests like \reffig{fig:lb-sdep}, introducing unsoundness when applied to languages like C++: compiling \reffig{fig:lb-sdep} for an ARM processor produces code that does exhibit the relaxed behaviour.
The logic of Lundberg et al~\cite{Lundberg} correctly discards many of this spurious program ordering, but it does not handle concurrency.
We provide an operational semantics of C++ that solves the thin air problem, where prior attempts use simplifying assumptions like those of prior logics~\cite{DBLP:conf/oopsla/NienhuisMS16,DBLP:conf/ecoop/DalvandiDDW19}.

\paragraph{Thin-air-free semantics.} 
Each of the concurrency definitions that solves the out-of-thin-air problem is remarkably complex~\cite{DBLP:journals/pacmpl/ChakrabortyV19,DBLP:conf/popl/KangHLVD17,DBLP:conf/pldi/LeeCPCHLV20,DBLP:journals/lmcs/JeffreyR19,ESOP2020-MRD,DBLP:conf/popl/Pichon-Pharabod16}, and so too is MRD. We choose to base our logic on MRD because MRD's semantic dependency relation hides much of the complexity of the model that calculates it. Where previously we would rely on program ordering, now we consider the semantic dependency provided by MRD.

\paragraph{Logics for thin-air-free models.} 
There are four logics built above concurrency models that solve the thin air problem: three of these are very simple and apply to only a handful of examples~\cite{DBLP:journals/pacmpl/JagadeesanJR20,DBLP:journals/lmcs/JeffreyR19,DBLP:conf/popl/KangHLVD17}, and one is a separation logic built above the Promising Semantics~\cite{DBLP:conf/esop/SvendsenPDLV18}. Unfortunately, the Promising Semantics allows unwanted out-of-thin-air behaviour, forbidden by MRD~\cite{ESOP2020-MRD}, and as a result may not support type safety~\cite{DBLP:journals/pacmpl/JagadeesanJR20}. 

\section{Conclusions and future work}
The subtle behaviour of concurrent programs written in optimised languages necessitates good support for reasoning, but existing logics make unsound assumptions that rule out compiler optimisations, or use underlying concurrency models that admit out of thin air behaviour (see Section~\ref{sec:related}). We present a logic built above MRD. MRD has been recognised as a potential solution to the thin air problem in C and C++ by the ISO~\cite{vote} and is the best guess at a C++ model that allows optimisation and forbids thin-air behaviours.

We follow a typical path for constructing a logic and start with an operational semantics that we show equivalent to MRD, and then as far a possible, follow the reasoning style of Owicki Gries. In each case, we diverge from a typical development because we cannot adopt program order into our reasoning system and instead must follow semantic dependency. Our logic supplants linear program order with a partial order, and where traditional pre- and post-conditions are indexed by the program counter, here we index them by their position in the partial order. This approach will work with any memory model that can provide a partially ordered structure over individual program actions.

The challenge in using the logic presented here is in managing the multitude of proof obligations that follow from the branching structure and lack of order in semantic dependency. This problem represents an avenue for further work: it may be possible to obviate the need for some of these additional proof obligations within the logic or to provide tools to more conveniently manage them.



%% file: appendix1.tex
\section{Operational Semantics of Doherty et al.}

\label{sec:full-op-sem}

This section presents the operational semantics from \cite{DBLP:conf/ppopp/DohertyDWD19}. We formally model parallel composition as a function $T \to \Comm$, where $T$ is the type of a thread identifier. The uninterpreted semantics is given in \reffig{fig:comm-sem} and the interpreted semantics in \reffig{fig:int-comm-sem-sample}. 

\begin{figure}[t]
  \centering \small

  $ \inference{ n = eval(e, \rffun) \qquad a = \levW{i}{x}{n}}{(i\!: [x] := e, \rffun)
    \whilestep{a}_t (\kwskip, \rffun) } $
  \qquad
  $ \inference{n = eval(e, \rffun) \qquad a = \levWR{i}{x}{n}}{(i\!: [x] :=^R e, \rffun)
    \whilestep{a}_t (\kwskip, \rffun) } $
  \medskip 
  
  $ \inference{ a = \levR{i}{x}{n}}{(i\!: r := [x], \rffun)
    \whilestep{a}_t (\kwskip, \rffun[r := n]) } $
  \qquad
  $ \inference{ a = \levRA{i}{x}{n}}{(i\!: r :=^A [x], \rffun)
    \whilestep{a}_t (\kwskip, \rffun[r := n]) } $
  \medskip 

  $ \inference{ n = eval(e, \rffun)}{(i\!: r := e, \rffun)
    \whilestep{a}_t (\kwskip, \rffun[r := n]) } $
    
        \medskip

  $\inference{}{(\kwskip ; C, \rffun) \whilestep{\tau}_t (C, \rffun)}$
  \ \
  $\inference{(C_1, \rffun) \whilestep{a}_t (C_1', \rffun')}{(C_1 ; C_2, \rffun) \whilestep{a}_t
    (C_1' ; C_2, \rffun')}$
  \medskip

  $ \inference
  {eval(B, \rffun) = \True}
  {(\kwif \ B\ \kwthen\
    C_1\ \kwelse\ C_2, \rffun) \whilestep{\tau}_t (C_1, \rffun)} $ 
  \ \
  $ \inference
  {eval(B, \rffun) = \False}
  {(\kwif\ B\ \kwthen\
    C_1\ 
    \kwelse\ C_2, \rffun) \whilestep{\tau}_t (C_2, \rffun)} $

  \medskip
  

  
$\text{\sc Prog}
  \inference{(P(t), \rffun) \whilestep{a}_t (C,\rffun')}
  {(P, \rffun) \whilestep{a}_t (P[t
    := C], \rffun')}$
    
  \caption{Uninterpreted operational semantics of
    commands and programs, where $\whilestep{\tau}$ represents an internal transition}
  \label{fig:comm-sem}
\end{figure}

\begin{figure}[t]
\begin{center}
$
\inference{P \whilestep{\tau}_t P'}{(P, (\sigma, \rffun))
  \ltsArrow{\tau} (P', (\sigma, \rffun))}$ 
\hfill
$\inference{(P, \rffun)
  \whilestep{a}_t (P', \rffun') \qquad b = (g, a, t) \qquad \sigma \strans{\ b\ } \sigma'
  } 
  {(P, (\sigma, \rffun)) \ltsArrow{b}
  (P', (\sigma', \rffun'))}
$
\vspace{-3mm}
\end{center}
  \caption{Interpreted operational semantics of programs}
  \label{fig:int-comm-sem-sample}
\end{figure}

\begin{example}
\label{ex:no-sdep-eg}
Consider the following program: \medskip

\hfill \begin{minipage}[b]{0.6\columnwidth}
\centering
\begin{tabular}[t]{l||@{\ }l}
\multicolumn{2}{c}{\tt Init: x = y = r1 = r2 =  0}\\
\begin{minipage}[t]{0.42\columnwidth}
\tt Thread 1 \\
1: r1 := [x] \\
2: y := 1  
\end{minipage}
&
\begin{minipage}[t]{0.52\columnwidth}
\tt Thread 2 \\
3: r2 := [x] \\
4: x := 1  
\end{minipage}
\end{tabular}\\[2pt]
\blue{\{$r1 \neq 1 \lor r2 \neq 1\}$}
\end{minipage} \hfill {}

\medskip

Consider the first tagged action graph below, which represents a possible state after executing lines the program above for the model in~\cite{DBLP:conf/ppopp/DohertyDWD19}, which is a model where $\ltsb$ and program order coincide. For clarity, we have omitted the tags and thread identifiers of each event. 

  \begin{center}
  \begin{minipage}{0.3\columnwidth}
       \scalebox{0.85}{
      \begin{tikzpicture}[node distance=.5cm]      
        \node (0) at (1.5,3) {$\levW{0}{x}{0},\levW{0}{y}{0}$};
        \node (1) at (0,1.5) {$\levR{1}{x}{0}$};
        \node (2) at (0,0) {$\levW{2}{y}{1}$};
        \node (3) at (3,1.5) {$\levR{3}{y}{1}$};
        \node (4) at (3,0) {$\levW{4}{x}{1}$};
        \path 
        (0) edge[hb] node[left] {$\ltsb$} (1)
        (0) edge[rf,bend right] node[left] {$\ltrf$} (1)
        (0) edge[hb] node[right] {$\ltsb$} (3)
        (2) edge[rf] node[left] {$\ltrf\ \ $} (3)
        (1) edge[hb] node[left] {$\ltsb$} (2)
        (0) edge[mo] node[right] {$\ltmo$} (2)
        (3) edge[hb] node[right] {$\ltsb$} (4)
        (0) edge[mo] node[left] {$\ltmo$} (4)
        ;  
      \end{tikzpicture}
    } 
  \end{minipage}
  \hfill
    \begin{minipage}{0.3\columnwidth}
            \centering 
                
                Derived $\lthb$ and $\ltfr$ \smallskip

       \scalebox{0.85}{
      \begin{tikzpicture}[node distance=.5cm]      
        \node (0) at (1.5,3) {$\levW{0}{x}{0},\levW{0}{y}{0}$};
        \node (1) at (0,1.5) {$\levR{1}{x}{0}$};
        \node (2) at (0,0) {$\levW{2}{y}{1}$};
        \node (3) at (3,1.5) {$\levR{3}{y}{1}$};
        \node (4) at (3,0) {$\levW{4}{x}{1}$};
        \path 
        (0) edge[hb] node[left] {$\lthb$} (1)
        (0) edge[hb] node[right] {$\lthb$} (3)
        (1) edge[hb] node[left] {$\lthb$} (2)
        (3) edge[hb] node[right] {$\lthb$} (4)
        (0) edge[hb] node[left] {$\lthb$} (4)
        (0) edge[hb] node[right] {$\lthb$} (2)
        (1) edge[fr] node[right] {\ \ $\ltfr$} (4)

        ;  
      \end{tikzpicture}
    } 
  \end{minipage}
    \hfill
      \begin{minipage}{0.3\columnwidth}
        \centering 
        Derived $\lteco$ \smallskip
        
      \scalebox{0.85}{
      \begin{tikzpicture}[node distance=.5cm]      
        \node (0) at (0,0) {$\levW{0}{x}{0}$};
        \node (1) at (1,1.25) {$\levR{1}{x}{0}$};
        \node (4) at (2,0) {$\levW{4}{x}{1}$};
        \path 
        (0) edge[rf] node[left] {$\ltrf$} (1)
        (0) edge[mo] node[above] {$\ltmo$} (4)
        (1) edge[fr] node[right] {$\ltfr$} (4)
        ;  
      \end{tikzpicture}
    } \medskip
    
  \scalebox{0.85}{
      \begin{tikzpicture}[node distance=.5cm]      
        \node (0) at (0,0) {$\levW{0}{y}{0}$};
        \node (2) at (2,0) {$\levW{2}{y}{1}$};
        \node (3) at (3,1.25) {$\levR{3}{y}{1}$};
        \path 
        (0) edge[mo] node[above] {$\ltmo$} (2)
        (2) edge[rf] node[right] {$\ltrf$} (3)
        ;  
      \end{tikzpicture}
    } 
      \end{minipage}
  \end{center}
  The derived happens-before and from-read relations are given in the second figure, and the derived extended coherence order in the third. 
\end{example}

\begin{example}
Suppose  $\sigma$ is the state depicted in \refex{ex:no-sdep-eg}. We have $\OW_\sigma(1) = \{\levW{0}{x}{0}, \levW{4}{x}{1}, \levW{2}{y}{1}\}$ (omitting the tag and thread id of the events) since thread $1$ has ``encountered''  $\levW{2}{y}{1}$, but not yet  encountered $\levW{4}{x}{1}$. Similarly, $\OW_\sigma(2) = \{\levW{2}{y}{1}, \levW{4}{x}{1}\}$ since it has encountered both  $\levW{2}{y}{1}$ and $\levW{4}{x}{1}$. 
\end{example}

%% file: appendix.tex
\section{Equivalence between MRD and the operational semantics}

\subsection{Operational executions are event-structure executions}
\begin{lemma}
\label{opsem-to-mrd}
Given some program $P$ such that 
$$(P, (\emptyset, \lambda x.\; (\lambda y. 0)), F) \fullrel{}^* (skip, (\sigma, \Gamma)), F')$$
we name the components of $\sigma$ to be $(D, sb, \ltrf_D, \ltmo)$.
There must exist an execution $X = (E, \dep, \ltrf_E, \ppo)$ in $\llbracket P \rrbracket_{n,\; (\lambda x.\; 0)\; \emptyset}$ and a bijection over events $\bij: E \rightarrow D$ such that $\bij(E) = D$ and $\bij(\ltrf_E) = \ltrf_D$ where $\bij(e) = (d, a, t) \implies \lambda(e) = a$.
\end{lemma}

Suppose the contrapositive: that for any execution $X$ in the MRD structure, either $D \not = \bij(E)$ or $\ltrf_D \not = \bij(\ltrf_E)$. 

If there are no executions containing exactly the actions in $D$, then we cannot arrive at event set $D$ by application of $\fullrel{a}$ - it requires $a \future f$ for at least one future, and the initial future set is derived directly from the set of executions. There must therefore exist at least one execution such that $D = \bij(E)$. Considering only these executions, we now show that at least one must have equivalent $\ltrf$ edges to $\ltrf_D$.

We know that MRD exhaustively generates $\ltrf$ edges which do not create a cycle in $\ltrf \cup \dep{} \cup \ppo$.
There cannot be an edge in all $\ltrf_E$ which is simply absent from $\ltrf_D$ with no new edge added, as this would imply that $\ltrf_D$ is incomplete. This is forbidden by the read execution rule.
To arrive at a complete $\ltrf_D$ relation which contains an edge not present in any potential $\ltrf_E$ where $E = D$, the proposed edge must cause such a cycle. Given that futures are ordered by $\dep \cup \ppo$, the new cyclical edge must connect a maximal event to a minimal one. Once again, the requirement that $e \future f$ for some $f$ prevents executing a maximal event before a minimal one.
Thus, there must be some $\ltrf_E$ which is equal to $\ltrf_D$.

\subsection{The operational semantics covers all event-structure executions.}
\newcommand{\rel}[1]{\xrightarrow{#1}}
\newcommand{\rf}{\xrightarrow{rf}} 
\begin{lemma}
\label{mrd-to-opsem}
Let $X = (E, \dep, rf, \ppo)$ be a complete execution of the MRD model $\llbracket P \rrbracket_{n\; \lambda x.\; 0,\; \emptyset}$ for large enough $n$ to ensure termination. Let the relation $\rel{R}$ be $\dep \cup rf \cup \ppo$, observing that $E$ is partially ordered by $\rel{R}^*$. 
There exists a relation $\rel{R'}$ disjoint from $\rel{R}$ such that $\rel{R \cup R'}$ is total and acyclic over $E$, giving the chain $e_1 \rel{R \cup R'} e_2 ... e_k$, and an execution of the operational semantics $\sigma_1 \strans{w, e_1} \sigma_2 ... \sigma_k$ such that $e_1 .. e_k \sim \sigma_k$, where $e_1 .. e_n$ denotes $X$ with the event set and relations restricted to $\{e_1 ... e_n\}$ and $(E, rf_E, dp, ppo) \sim (D, rf_D, mo)$ iff $\lambda(E) = D \wedge \lambda(rf_E) = rf_D$.
\end{lemma}

\paragraph{Proof sketch}
The existence of an acyclic $\rel{R \cup R'}$ is given by the acyclicity of $\rel{R}$ asserted by the MRD axiomatic check.
We proceed by induction over $\rel{R \cup R'}$ and $\strans{w, e}$.
The base case is trivial:
take any arbitrary $\rel{R'}$ which makes $\rel{R \cup R'}$ total. Because $e_1$ must be minimal in $\rel{R}$, it must be minimal in $\rf$ and minimal in $\dep$. Being minimal in $\rf$ means it cannot be a read or RMW event. Being minimal in $\dep$ means it must be in an available future. Let $\sigma_1$ be $(\{\lambda(e_1)\}, \emptyset, \emptyset)$.
It remains to show that $e_1 .. e_{n-1} \sim \sigma_{n-1} \implies e_1 .. e_n \sim \sigma_n$ for $n \leq k$.
There are 3 cases for the event type $e_n$: a read event, a write event, or an RMW event.

In the case where $e_n$ is a read, to use the semantic rule for creating state $\sigma_{n}$ requires an event $\lambda(w) \in OW_{\sigma_{n-1}}(tid(e_n))$ such that $var(w) = var(e_n) \wedge wrval(w) = wrval(e_n)$. This adds $(\lambda(w), \lambda(e_n))$ to $rf_D$.

Given $X$ is complete, there must be some write $w'$ such that $w' \rf e_n$. It must also be true that $w' \rel{R \cup R'} e_n$, due to the inclusion of $rf$ in $\rel{R}$. Therefore, we can use $w'$ as $w$.

The event $w$ will be in $OW_{\sigma_{n-1}}(tid(w))$ if it is maximal in $mo$. 
Suppose some event $w''$ exists such that $(w, w'') \in mo$. By the construction of $mo$, it must also be true that $w \rel{R \cup R'}^* w''$. If $w \rel{R}^* w''$, then the axiomatic checks in MRD prevent the $w \rf e_n$ edge from being added. Otherwise, by lemma \ref{new_arb_rel}, there exists another choice of $\rel{R'}$ such that $w'$ is added to $\sigma_n$ before $w$ and $(w, w')$ is no longer maximal in $mo$.

In the case where $e_n$ is an initialising write, we only need $D$ such that $e_n \in D$.
If $e_n$ is a non-initialising write, meaning there is some event $w \in X$ which is a write to the same variable as $e_n$, we must show that $w$ is in $OW(e_n)$. This follows by the same reasoning as in the read case.

\begin{lemma} \label{new_arb_rel}
If $e_i \rel{R \cup R'}^* e_j \wedge e_i \not \rel{R}^* e_j$ in the chain $e_1 .. e_n$ then there exists a valid choice of relation $\rel{R''}$ such that $e_j \rel{R \cup R''}^* e_i$, the union $\rel{R \cup R''}$ is acyclic, and
there exists an event $e_n'$ in $e_1 .. e_n$ such that
\begin{enumerate}
    \item $e_1 .. e_n$ under $\rel{R \cup R'}$ is equal to $e_1 .. e_n'$ under $\rel{R \cup R''}$
    \item $\sigma_n'$, the result of running the operational semantics under $\rel{R \cup R''}$, is equal to $\sigma_n$ excepting relation edges between $e_i$ and $e_j$
\end{enumerate}
\end{lemma}
\begin{proof}

Begin by enforcing that $e_j \rel{R''} e_i$. A cycle in $\rel{R \cup R''}$ 
would have the shape $e_i \rel{R \cup R''} e_j$. As we can trivially prevent $e_i \rel{R''}^* e_j$ by construction, this relation edge must contain one or more event pairs $e$ and $e'$ such that $e \rel{R}^* e'$. This event pair can be moved $\rel{R''}$-before $e_j$ unless $e_j \rel{R}^* e$, and it can be moved $\rel{R''}$-after $e_i$ unless $e' \rel{R}^* e_i$. If both of these are true, then $e_i \rel{R}^* e_j$, which we know is untrue.

To find $e_n'$, choose any element which is not equal to $e_j$ from the set of events which are maximal in $\rel{R}$ and make it maximal in $\rel{R''}$. There must be some element in this set which is not $e_j$, as otherwise we would have $\forall e. e \rel{R} e_j$ and it would again be true that $e_i \rel{R}^* e_j$.

What remains is to create two arbitrary orderings $R''_1$ and $R''_2$ such that $e_1 \rel{R \cup R''_1} e_j$ and $e_j \rel {R \cup R''_2} e_n$ and both $\rel{R \cup R''_1}$ and $\rel{R \cup R''_2}$ are acyclic. Partition all remaining events into those which are and are not $\rel{R}$-before $e_j$, then construct $\rel{R''_1}$ over the first and $\rel{R''_2}$ over the second. These are individually trivial by the acyclicity of $\rel{R}$, and their union must be acyclic as their event sets are disjoint.

It is trivial to establish that $e_1..e_n = e_1..e_{n'}$ by the fact that all relation edges in $X$ are preserved in $\rel{R}$. 
To show that $\sigma_n$ is equivalent to $\sigma_{n'}$ excepting $(e_i, e_j)$ edges, it suffices to show that if $(e, e') \in mo$ then $(e, e') \in mo'$, as rf edges are constrained by $\rel{R}$.
Events are only relocated around points $e_i$ and $e_j$, so any such edges would have the shape $(e, e_j)$ or $(e_i, e)$ in $\sigma_n$ and become $(e_j, e)$ or $(e, e_i)$ respectively.
Focusing on the first case, any such reordering is only \emph{required} if $e_j \rel{R}^* e$, as this forces $e$ to be placed $\rel{R \cup R''}^*$-after $e_j$ in our construction of $\rel{R''}$. This makes the initial mo edge impossible to observe. 
The same reasoning holds for the $(e_i, e)$ edge.

\end{proof}

%% file: appendix-example.tex
\newpage

\section{Further Examples}

\subsection{Partial conditional}

Our first example is the partial conditional example, whose proof outline is given in \reffig{fig:par-cond}. The program contains no dependencies in Thread 1 and non-trivial dependencies in thread 2. In particular, in Thread 2 the dependencies are condition on the values read for $x$ and $y$. 
\begin{itemize}
    \item If $r1 = r2 = 1$, then we have dependencies between lines 3 and 5, and between lines 4 and 6, but lines 3/4 and 5/6 are mutually independent. 
    \item If $r1 = 1$ and $r2 = 0$, there is a dependency between lines 3 and 7. 
    \item If $r1 = 0$ and $r2 = 1$, there is a dependency between lines 4 and 8. 
\end{itemize}
This is reflected by the futures $F$ for thread 1 and $G$ for thread 2 in \reffig{fig:par-cond}.

\begin{figure}[h]
    \centering
\begin{tabular}[t]{r||@{\ \ }l}
\multicolumn{2}{l}{{\tt Init: x = y = r1 = r2 = 0}}\\
\multicolumn{1}{l}{\tt Thread 1} & 
\multicolumn{1}{l}{\tt Thread 2} \\
\begin{minipage}[t]{0.36\columnwidth}\tt
 \vspace{-0.75em}
 \blue{$\{[x = 0]_2 \wedge [y = 0]_2\}_F$} \\
 \phantom{xx}\blue{\textrm{$\{[x = 0]_2 \wedge [y \in \{0,1\}]_2\}_{F_2}$}} \\
\phantom{xx}1: [x] := 1 \\
  \phantom{xx}\blue{\textrm{$\{[x \in \{0,1\}]_2 \wedge [y =0]_2\}_{F_1}$}} \\
\phantom{xx}2: [y] := 1  \\
\blue{$\{[x \in \{0,1\}]_2 \wedge [y \in \{0,1\}]_2\}_\emptyset$}
\end{minipage}
&
\begin{minipage}[t]{0.64\columnwidth} \tt 
 \vspace{-0.75em}

 \blue{\textrm{$\{I\}_G$}} \\
  \phantom{xx}\blue{\textrm{$\{r2 = m \wedge I\}_{G_{4_n}}$}}\\
 \phantom{xx}3: \ r1 := [x] ;\\
  \phantom{xx}\blue{\textrm{$\{r1 = m \wedge I\}_{G_{3_m}}$}}\\
 \phantom{xx}4: \ r2 := [y]  ; \\
 \phantom{xx}\blue{\textrm{$\{r1 = m 
 \wedge r2 =n \wedge I\}_{G_{3_m,4_n}}$}}\\
\phantom{xx}\phantom{4:}  \ {\bf if}\ r1 = 1 \& r2 = 1  \{ \\
 \phantom{xxxxx} \blue{\textrm{$\{r1 = 1 \wedge r2 = 1\wedge I\}_{G_{3_1,4_1}}$}}
 \\
  \phantom{xxxxxxxx} \blue{\textrm{$\{r1 = 1 \wedge r2 = 1\wedge [w=0]_2 \wedge [z=1]_2\}_{G_{3_1,4_1,6}}$}}\\
\phantom{xxxxxxxx} 5: \ [w] := 1  ; \\
 \phantom{xxxxxxxx} \blue{\textrm{$\{r1 = 1 \wedge r2 = 1\wedge [w=1]_2 \wedge [z=0]_2\}_{G_{3_1,4_1,5}}$}}\\
 \phantom{xxxxxxxx} 6: [z] := 1  ; \} \\
 \phantom{xx}\phantom{4:} \ {\bf if}\ r1 = 0 \& r2 = 1 \{  \\
 \phantom{xxxxx} \blue{\textrm{$\{r1 = 0 \wedge r2 = 1\wedge I\}_{G_{3_0,4_1}}$}}\\
 \phantom{xxxxxxxx} 7:  z := 1  \} \\
\phantom{xx}\phantom{4:} \  {\bf if}\ r1 = 1 \& r2 = 0 \\
 \phantom{xxxxx} \blue{\textrm{$\{r1 = 1 \wedge r2 = 0\wedge I\}_{G_{3_1,4_0}}$}}\\
\phantom{xxxxxxxx}8: w := 1  ;  \\
 \blue{\{$[w = r1]_2 \wedge [z = r2]_2\}_\emptyset$}
\end{minipage}
\\
\multicolumn{2}{c}{\blue{\{$[w = r1]_2 \wedge [z = r2]_2\}$\qquad \qquad \qquad \qquad }}\\
\end{tabular}
\caption{Partial conditional: $F = \{\{1, 2\}\}$ $G = \{
\{3_1 \prec 5, 4_1 \prec 6\}, \{3_1 \prec 8, 4_0\}, \{4_1 \prec 7, 3_0\}, \{3_0, 4_0\}\}$ and $I= [w = 0]_2 \wedge [z = 0]_2$}
    \label{fig:par-cond}
\end{figure}



    
 


\subsection{Random number generator (RNG)}

Our next example is the random number generator (RNG) litmus test (see \reffig{fig:rng}) from \cite{DBLP:journals/pacmpl/ChakrabortyV19}, which in a poorly defined semantics has the potential to terminate so that line~7 is executed. Our proof outline shows that this cannot be the case in our semantics, i.e., the program satisfies the postcondition $[x \not\approx 99]_{1,2,3}$.

\begin{figure}[h]

\centering

\noindent
\begin{tabular}[t]{l||@{\ \ }l@{\ \ }||@{\ \ }l}
\multicolumn{3}{l}{{\tt Init:} $x = y = r1 = r2 = r3 = 0$}\\
\multicolumn{1}{l}{\tt Thread 1} & 
\multicolumn{1}{l}{\tt Thread 2} & 
\multicolumn{1}{l}{\tt Thread 3} \\
\begin{minipage}[t]{0.29\columnwidth}\tt
 \vspace{-0.75em}
 \blue{\textrm{$\{[y = 0]_2 \wedge r2  = 0 $}}
\\ 
\blue{\textrm{${}\wedge \forall i.\ i \neq 0 \imp  [x \not\approx i]_1\}_F$}} \\
1: r1 := x \\
 \blue{\textrm{$\{[y = 0]_3 \wedge r1 = 0\}_{F_1}$}}\\
2: y := r1 + 1  \\
\blue{$\{true\}_\emptyset$} 
\end{minipage}
&
\begin{minipage}[t]{0.33\columnwidth}\tt
 \vspace{-0.75em}
 \blue{\textrm{$\{[x = 0]_{1,2,3}$}} \\
\ \ \ \ \blue{\textrm{${} \wedge \forall i.\ i \notin \{0,1\} \imp [y \not\approx i]_2\}_{G}$}} \\
3: r2 := y \\
 \blue{\textrm{$\{[x = 0]_{1,2,3} \wedge r2 \in \{0,1\}\}_{G_3}$}}\\
4: x := r2  \\
\blue{$\{true\}_\emptyset$}
\end{minipage}
&
\begin{minipage}[t]{0.45\columnwidth} \tt 
 \vspace{-0.75em}
\blue{\{$\forall i \notin \{0,1\} \imp [x \not \approx i]_3$} \\
\blue{${} \wedge [x \not\approx 99]_{1,2,3}\}_H$} \\
 5: r3 := x \\
 \blue{\textrm{$\{r3 \in \{0,1\} \wedge [x \not\approx 99]_{1,2,3}\}_{H_5}$ }} \\
\phantom{5:} \kwif\, r3 = 100 \kwthen \\
\blue{\ \ \ \ \ \ $\{\False\}_{H_5}$} \\
6: \ \ \ x := 99\\
\blue{$\{[x \not\approx 99]_{1,2,3}\}_\emptyset$}
\end{minipage}
\end{tabular}\\[2pt]
\blue{\{$[x \not\approx 99]_{1,2,3}$\}} 
\caption{RNG, where $F=\{\{1 \prec 2\}\}$, $G=\{\{3 \prec 4\}\}$, $H=\{\{5 \prec 6\}\}$}
\label{fig:rng}
\end{figure}

The program is similar to \reffig{fig:lb-sdep-po}, but contains semantic dependencies between lines 1 and 2, between lines 3 and 4, and between lines 5 and 7. This means that no reorderings are permitted within any of the threads, thus its proof is straightforward. In fact, we have used our existing Isabelle/HOL formalisation~\cite{DBLP:journals/corr/abs-2004-02983} for a model that assumes program order to discharge the proof obligations.

\begin{theorem}
The proof outline in \reffig{fig:rng} is valid. 
\end{theorem}

\subsection{LB+data-add+ctrl}

The next two examples are designed to test whether a chain of semantic dependencies on data, addresses and control are properly modelled by the semantics.

\begin{figure}[h]

\centering

\noindent
\begin{tabular}[t]{l||@{\ \ }l@{\ \ }}
\multicolumn{2}{l}{{\tt Init: x = y = r1 = r2 = 0}}\\
\multicolumn{1}{l}{\tt Thread 1} & 
\multicolumn{1}{l}{\tt Thread 2} 
\\
\begin{minipage}[t]{0.35\columnwidth}\tt
 \vspace{-0.75em}
 \blue{\textrm{$\{[y = 0]_1 \wedge r2  = 0 $}}
\\ 
\blue{\textrm{${} \wedge \forall i.\ i \neq 0 \imp  [x \not\approx i]_1\}_F$}} \\
1: r1 := x \\
 \blue{\textrm{$\{[y = 0]_2 \wedge r1 = 0 \wedge r2 = 0\}_{F_1}$}}\\
2: z := r1 + 1  \\
\blue{\{$r1 \neq 1 \lor r2 \neq 2\}_\emptyset$} 
\end{minipage}
&
\begin{minipage}[t]{0.34\columnwidth}\tt
 \vspace{-0.75em}
\blue{\textrm{$\{\forall i.\ i \notin \{0,1\} \imp [y \not\approx i]_2\}_{G}$}} \\
3: r2 := z \\
 \blue{\textrm{$\{r2 \in \{0,1\}\}_{G_3}$}}\\
\phantom{5:} \kwif\ r2 = 2\ \kwthen \\
\phantom{5:} \quad \blue{$\{{\it false}\}_{G_3}$} \\
4: \quad x := 1 \\
\blue{\{$true\}_\emptyset$}
\end{minipage}
\end{tabular}\\[2pt]
\blue{\{$r1 \neq 1 \lor r2 \neq 2$\}} 
\caption{Simple LB+data-add+ctrl, where $F= \{\{1\prec 2\}\}$ and $G = \{\{3\prec4\}\}$ }
\label{fig:simp-LB+data-add+ctrl}
\end{figure}

\begin{figure}[h]

\centering

\noindent
\begin{tabular}[t]{l||@{\ \ }l@{\ \ }}
\multicolumn{2}{l}{{\tt Init: x = y = r1 = r2 = r3 = r4 = 0}}\\
\multicolumn{1}{l}{\tt Thread 1} & 
\multicolumn{1}{l}{\tt Thread 2} 
\\
\begin{minipage}[t]{0.35\columnwidth}\tt
 \vspace{-0.75em}
\blue{\textrm{\{$[y = 0]_1 \wedge [z = 0]_2$}}
\\
\blue{\textrm{$\wedge (\forall i.\ i \neq 0 \imp  [x \not\approx i]_1\}_{F}$}} 
\\
1: r1 := x 
\\
 \blue{\textrm{$\{[y = 0]_1 \wedge [z = 0]_2 \wedge r1  = 0\}_{F_1}$}}\\
2: r2 := r1; \\
\blue{\{$[y = 0]_1 \wedge [z = 0]_2 \wedge r2  = 0\}_{F_2}$}  \\
3: y := r2; \\
\blue{\{$[y = 0]_1 \wedge [z = 0]_2\}_{F_3}$}  \\
4: r3 := y; \\
\blue{\{$[z = 0]_2 \wedge r3  = 0\}_{F_4}$}  \\
5: z := r3 + 1; \\
\blue{\{$r3 = 0\}_\emptyset$}  \\

\end{minipage}
&
\begin{minipage}[t]{0.5\columnwidth}\tt
 \vspace{-0.75em}
\blue{\textrm{$\{\forall i.\ i \notin \{0,1\} \imp [z \not\approx i]_2\}_G$}} \\
6: r4 := z \\
 \blue{\textrm{$\{r4 \in \{0,1\}\}_{G_6}$}}\\
\phantom{xx} \kwif\ r4 = 2\ \kwthen \\
\phantom{5:} \quad \blue{$\{{\it false}\}_{G_6}$} \\
7: \quad x := 1 \\
\blue{$\{r4 \in \{0,1\}\}_{\emptyset}$}
\end{minipage}
\end{tabular}\\[2pt]
\blue{$\{r3 \neq 1 \lor r4 \neq 2\}$} 
\caption{LB+data-add+ctrl, where $F= \{\{1\prec 2 \prec 3 \prec 4 \prec 5\}\}$ and $G = \{\{6\prec7\}\}$}
\label{fig:LB+data-add+ctrl}
\end{figure}

The proofs of both programs are straightforward since they do not contain any intra-thread re-orderings. Again, we have shown correctness of both proof outlines using our existing Isabelle/HOL formalisation~\cite{DBLP:journals/corr/abs-2004-02983}.

\begin{theorem}
The proof outlines in \reffig{fig:simp-LB+data-add+ctrl} and \reffig{fig:LB+data-add+ctrl} are valid. 
\end{theorem}
